\documentclass[aps,reprint,twocolumn,superscriptaddress,english,prl,showpacs,longbibliography]{revtex4-2}

\usepackage{amsmath,amsthm,amssymb}
\usepackage{graphicx}

\usepackage{dcolumn}
\usepackage{bm}
\usepackage{color}
\usepackage{multirow}
\usepackage{mathrsfs}
\usepackage{hyperref}
\usepackage{cleveref}
\usepackage{autobreak}

\usepackage{array}
\usepackage{makecell}
\usepackage{tabularx}
\usepackage{booktabs}
\usepackage{threeparttable}
\usepackage[ruled,vlined]{algorithm2e}

\newcommand{\braket}[2]{\langle#1|#2\rangle}

\newcommand{\ket}[1]{\left|#1\right\rangle}

\newtheorem{proposition}{Proposition}
\newtheorem{definition}{Definition}
\newtheorem{theorem}{Theorem}

\usepackage{listings}
\usepackage{xcolor}

\definecolor{juliagreen}{HTML}{389826}

\lstdefinestyle{JuliaREPL}{
  basicstyle=\ttfamily\small,
  backgroundcolor=\color{gray!10},
  frame=single,
  columns=fullflexible,
  showstringspaces=false,
  literate={julia> }{{{\color{juliagreen}julia> }}}1
}
\crefname{algocf}{algorithm}{algorithms}
\Crefname{algocf}{Algorithm}{Algorithms}

\usepackage{geometry}
\usepackage{multirow}
\begin{document}

\title{Encoding computationally hard problems in triangular Rydberg atom arrays}

\author{Xi-Wei Pan}
\affiliation{Thrust of Advanced Materials, The Hong Kong University of Science and Technology (Guangzhou), Guangdong, China}
\author{Huan-Hai Zhou}
\affiliation{Thrust of Advanced Materials, The Hong Kong University of Science and Technology (Guangzhou), Guangdong, China}
\author{Yi-Ming Lu}
\affiliation{Zhili College, Tsinghua University, Beijing, China}
\author{Jin-Guo Liu}
\email{jinguoliu@hkust-gz.edu.cn}
\affiliation{Thrust of Advanced Materials, The Hong Kong University of Science and Technology (Guangzhou), Guangdong, China}

\date{\today}

\begin{abstract}
    Rydberg atom arrays are a promising platform for quantum optimization, encoding computationally hard problems by reducing them to independent set problems with unit-disk graph topology. In \href{https://doi.org/10.1103/PRXQuantum.4.010316}{[Nguyen et al., PRX Quantum 4, 010316 (2023)]}, a systematic and efficient strategy was introduced to encode multiple problems into a special unit-disk graph: the King's subgraph.
    However, King's subgraphs are not the optimal choice in two dimensions. Due to the power-law decay of Rydberg interaction strengths, the approximation to unit-disk graphs in real devices is poor, necessitating post-processing that lacks physical interpretability.
    In this work, we develop an encoding scheme that can universally encode computationally hard problems on triangular lattices, based on our innovative automated gadget search strategy.
    Numerical simulations demonstrate that quantum optimization on triangular lattices reduces independence-constraint violations by approximately two orders of magnitude compared to King's subgraphs, substantially alleviating the need for post-processing in experiments. 
\end{abstract}

\maketitle

\paragraph{Introduction.}— Classical algorithms for NP-hard optimization problems scale exponentially~\cite{schrijver2003combinatorial,bernhard2008combinatorial,Moore2011}, motivating the search for quantum approaches with potential advantages~\cite{farhi_quantum_2000,Farhi_2001,Kadowaki_1998,Das_2008,lidar_albash_review_2018,farhi_quantum_2014,Lucas_Ising_Formulation}. Rydberg atom arrays are a promising platform for quantum optimization, as blockade interactions directly enforce independence constraints~\cite{saffman2010quantum,Pichler2018MIS,barik2024quantum,wurtz2023aquila,browaeys2020many}. This naturally maps to the maximum independent set (MIS) problem on unit-disk graphs—graphs where vertices are connected if separated by less than a fixed distance—an NP-hard problem~\cite{pichler2018computational,kim2024quantum,schuetz2025qredumis}. With tunable local detunings, this mapping extends to the maximum weighted independent set (MWIS), also NP-hard~\cite{de2025demonstration,schuetz2025qredumis}. Recent experiments have even reported signatures of superlinear quantum speedup for MIS on King's subgraphs (KSGs)—a restricted subclass of unit-disk graphs—with up to 289 qubits~\cite{Ebadi2022}, highlighting the central role of KSG embeddings in current Rydberg-atom implementations.

To make Rydberg atoms useful for practical optimization problems, Nguyen et al.~\cite{Nguyen2023} introduced the use of KSGs to encode independent set problems with arbitrary connectivity. The key idea is to use \emph{gadgets}, small atomic arrangements that locally enforce constraints while respecting unit-disk connectivity.
This approach has a provably optimal vertex overhead of $O(n^2)$, where $n$ is the number of vertices in the original graph.

The main critique of KSG-based quantum optimization~\cite{Bombieri2025,Cazals2025,cazals2025quantum,Ebadi2022} is that the approximation to unit-disk graphs is poor, leading to the need for post-processing that lacks explainability.
In Ref.~\cite{Ebadi2022}, a significant portion of independence constraints were observed to be violated, with many output configurations having cardinalities exceeding that of the maximum independent set. Extensive post-processing, including removing vertices that violate the independence constraint and adding new vertices greedily, is necessary.
Partly due to this, the observed superlinear speedup is not completely convincing.  

A convincing demonstration of quantum speedup should rest on a robust encoding scheme that avoids post-processing, which could be achieved by changing the lattice type, e.g., to a triangular lattice.
In this letter, we introduce a key insight: the \emph{quality factor} $Q = R_{\text{min}}/r_{\text{max}}$ quantifies encoding robustness, where $R_{\text{min}}$ is the minimum distance between non-adjacent atoms and $r_{\text{max}}$ is the maximum distance between adjacent atoms. 
Since Rydberg interactions between atoms separated by distance $r$ scale as $V(r)\propto 1/r^6$, the energetic separation between connected and non-connected pairs scales as $Q^6$.
KSGs yield only $Q = \sqrt{2}$, corresponding to an energetic separation of $Q^6=8$, 
whereas triangular lattices naturally provide $Q = \sqrt{3}$, giving $Q^6=27$—more than three times larger (see Supplemental Material for details). 
This enhanced penalty substantially suppresses constraint violations and mitigates unwanted long-tail effects~\cite{Cazals2025}.
The higher quality factor is also widely believed to be beneficial for quantum simulation of exotic phases~\cite{Zeng2025,Patil2025,Li2022}.

Yet, a general encoding scheme for triangular-lattice subgraphs (TLSGs) is lacking. 
Here TLSGs refer to unit-disk graphs constrained to the triangular lattice. 
It is even unknown whether MWIS on TLSGs is NP-hard~\cite{Moore2011}, 
casting doubt on universal encodability. A key obstacle is the absence of a systematic gadget-search framework. In particular, gadgets such as the crossing gadget serve as fundamental building blocks in encoding constructions. Ref.~\cite{Nguyen2023} identified an optimal 8-vertex crossing gadget for KSGs, but the brute-force search employed there quickly becomes infeasible for larger gadgets. For TLSGs, no valid crossing gadget had been reported until now.

In this letter, we develop an innovative automated gadget-search strategy, discover a 12-vertex crossing gadget for TLSGs, construct a universal encoding for NP-hard problems using this gadget, and numerically show that independence-constraint violations are reduced by nearly two orders of magnitude compared with KSG encodings. These results suggest that post-processing may be unnecessary for TLSG encodings.

\begin{figure*}[htbp]
    \centering
    \includegraphics[width=0.95\linewidth]{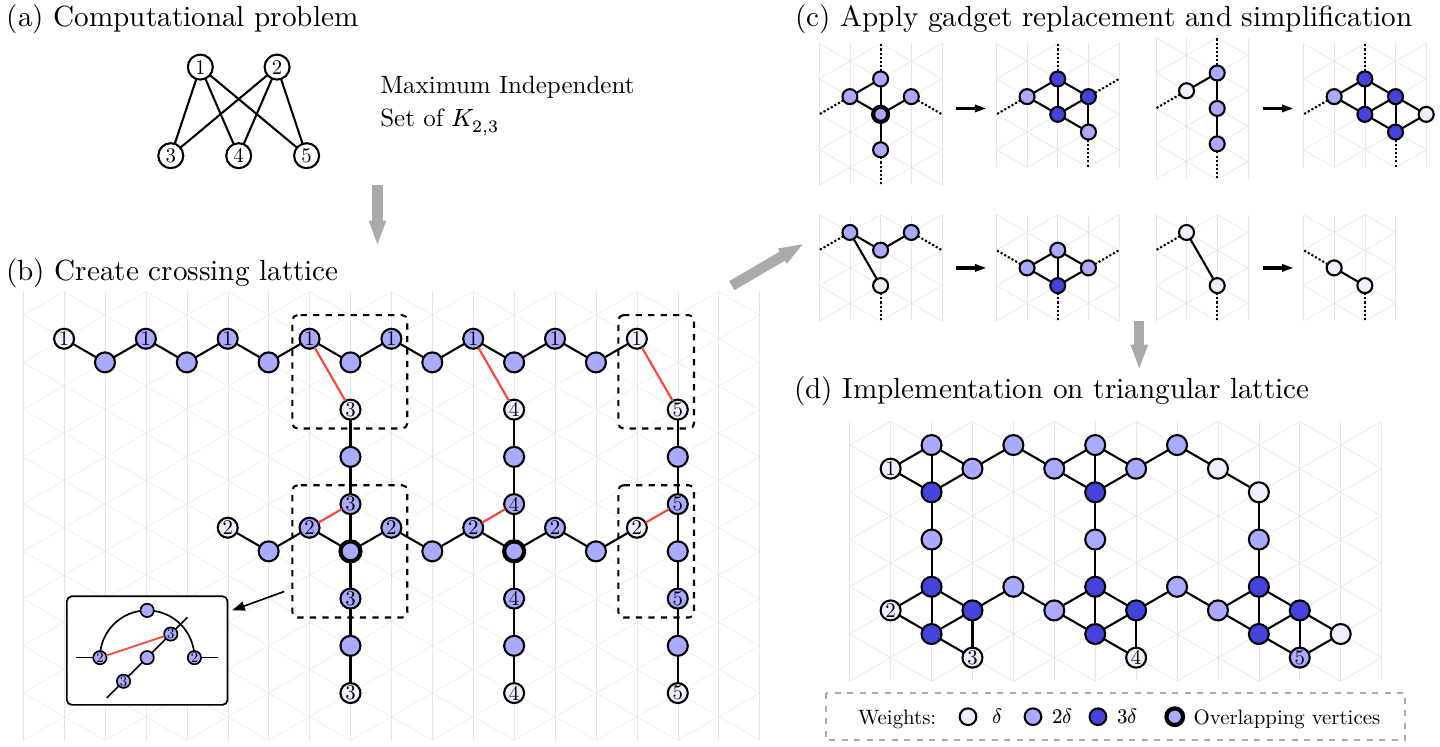}
    \caption{Procedure for encoding a MWIS/MIS problem on an arbitrary graph into MWIS on a TLSG, enabling optimization using programmable Rydberg atom arrays.
    (a) Example problem: finding MIS of $K_{2,3}$ graph.
    (b) Extend each vertex to a copy gadget and form a crossing lattice; numbers label equivalent copies of the same source vertex, and red edges indicate the original-graph connections.
    (c) Replace the substructures that violate the unit-disk constraint by logically equivalent gadgets. The gadgets and their composition rule are detailed in Fig.~\ref{fig:show_gadget}.
    (d) Final encoding of $K_{2,3}$; the solution to the original problem is inferred from the ground state configuration of the numbered nodes.
    }
    \label{fig:main}
\end{figure*}

\paragraph{Encoding scheme.}—
The MIS or MWIS problem on unit-disk graphs can be cast as a ground-state search for an effective Hamiltonian implemented by Rydberg atom arrays~\cite{Nguyen2023,de2025demonstration}. The blockade interaction enforces the independence constraint by forbidding simultaneous excitations of nearby atoms, while detunings encode vertex weights. Let $\Delta_v$ denote the detuning of atom $v$ and $R_b$ the blockade radius. The effective classical Hamiltonian is
\begin{equation}
H = -\sum_v \Delta_v n_v + \sum_{|\mathbf{r}_{u} - \mathbf{r}_{v}| < R_b} \infty\, n_u n_v,
\end{equation}
where $n_v$ is the number operator of atom $v$, and $\mathbf{r}_u$ and $\mathbf{r}_v$ are the atomic positions. The infinite interaction crudely approximates the fast-decaying van der Waals potential $\sim 1/r^6$ and enforces the hard constraint.

Since many computational problems can be reduced to MIS or MWIS on graphs~\cite{gao2025programming}, a central challenge is encoding graphs with arbitrary connectivity into hardware-native unit-disk geometries. In this letter, we focus on the TLSG and propose an efficient encoding framework. 

Consider a general graph \(G = (V, E)\) with vertex set $V$ and edge set $E$.  

\begin{theorem}
    The problem of finding a maximum independent set on a general
graph $G= (V,E)$ can be encoded into that on a TLSG with $O(|V|^2)$ vertices.
\label{thm:main}
\end{theorem}

\begin{proof}
We prove this theorem by constructing a two-step mapping scheme, illustrated in Fig.~\ref{fig:main}. In Step 1, each logical vertex of the source graph \(G=(V,E)\) is replaced by a vertex wire consisting of $O(|V|)$ vertices, corresponding to the copy gadget in Ref.~\cite{Nguyen2023}. The collection of these vertex wires forms a two-dimensional \emph{crossing lattice}, where crossings occur at specific sites for any $(u,v) \in E$. Vertices with odd indices along the same wire represent equivalent copies of the same source vertex, allowing an edge in the source graph to be redistributed to any pair of equivalent vertices. The redistributed edges, shown as red lines in Fig.\ref{fig:main}(b), are placed at the crossings of the lattice.

In Step 2, we use specialized gadgets with constant size to replace crossings in the crossing lattice (Fig.~\ref{fig:main}(c)), ensuring that all vertices respect the unit-disk constraint, while preserving the solution equivalence. 
Details on how to design such gadgets are discussed in the following section.
The generated graph can be further simplified, e.g., by trimming the dangling legs. The resulting TLSG encoding is shown in Fig.~\ref{fig:main}(d), with further details provided in the Supplemental Material.
\end{proof}

For a graph with $n$ vertices, our TLSG encoding scheme introduces an overhead of at most $O(8n^2)$ vertices (see Supplemental Material for details), somewhat higher than the $O(4n^2)$ reported for KSG encodings~\cite{Nguyen2023,schuetz2025quantum}, but still quadratic. Potential optimizations, such as vertex reordering~\cite{Nguyen2023} and further simplifications, can reduce the required number of atoms.

\paragraph{Gadget finding.}
— This section introduces the gadget search methodology, applicable to graphs both with and without geometric constraints, including, but not limited to, TLSGs.

A gadget is a small positive-weighted graph whose MWIS ground states realize a prescribed logical relation among a selected subset of vertices. Such relation can be expressed as a $k$-variable logical constraint with satisfying set $\mathcal{L} \subseteq \{0,1\}^k$. For example, a 2-to-1 AND gadget encodes the constraint $\mathcal{L} = \{000, 010,100,111\}$.

Gadgets are basic building blocks in our encoding scheme due to their \emph{composability}~\cite{Nguyen2023}: 
MWIS ground states of multiple gadgets can be combined to realize conjunctions of multiple logical constraints (see Fig.~\ref{fig:show_gadget}(d)).  

\begin{definition}[Logical equivalence]
    Let $G_\mathcal{L}=(V,E,\Delta)$ be a weighted graph with positive vertex weights 
    $\Delta = \{\Delta_v\}_{v\in V}$, and let $P = \{p_1, p_2, \dots, p_k\} \subseteq V$ 
    be an ordered set of \emph{pin vertices}. 
    
    For each MWIS $M \subseteq V$, define the pin-projection 
    \[
    \pi_P(M)_i =
    \begin{cases}
    1, & p_i \in M, \\
    0, & p_i \notin M,
    \end{cases}
    \quad i=1,\dots,k .
    \]
    The set of projected MWIS configurations is
    \[
    \mathcal{M}_{\mathrm{MWIS}}^P(G_\mathcal{L})
    := \{\pi_P(M)\mid M\in \mathcal{M}_{\mathrm{MWIS}}(G_\mathcal{L})\}
    \subseteq \{0,1\}^k .
    \]
    
    We say the MWIS solutions of $G_\mathcal{L}$ are \emph{logically equivalent on $P$} 
    to a $k$-variable logical constraint $\mathcal{L}\subseteq \{0,1\}^k$ if
    \[
    \mathcal{M}_{\mathrm{MWIS}}^P(G_\mathcal{L}) \;=\; \mathcal{L}.
    \]
    \end{definition}
    
    \begin{definition}[Gadget]
        A \emph{gadget} for a logical constraint $\mathcal{L}$ is a positive-weighted graph 
        $G_\mathcal{L}=(V,E,\Delta)$ with designated pins $P\subseteq V$ such that 
        its MWIS solutions are logically equivalent to $\mathcal{L}$ on $P$. 
        In particular, the gadget must be \emph{non-degenerate}, meaning that each 
        assignment in $\mathcal{L}$ corresponds to exactly one MWIS solution.
    \end{definition}

\begin{proposition}[MWIS is maximal]
    \label{thm:mwis}
    For any positive-weighted graph $G = (V, E, \Delta)$, a MWIS is always \emph{maximal}; that is, no additional vertex can be added without violating independence.
\end{proposition}
\begin{proof}
    If a MWIS $M$ were not maximal, there would exist a vertex $u \in V\setminus M$ with no neighbors in $M$. Then $M \cup \{u\}$ is independent and has strictly larger weight than $M$, contradicting maximality. 
    Hence, every MWIS is maximal.
\end{proof}

Given a target logical constraint $\mathcal{L}\subseteq\{0,1\}^k$ and a candidate 
graph space $\mathcal{G}$ (e.g., unit-disk graphs on a triangular lattice), we assign 
a designated set of vertices $P=\{p_1,\ldots,p_k\}\subseteq V$ in each candidate 
graph $G \in \mathcal{G}$. For each candidate graph, we enumerate its 
\emph{maximal} independent sets $\mathcal{M}$ (see \Cref{thm:mwis}) and extract the subset 
$\mathcal{M}_{\min}\subseteq \mathcal{M}$ that is sufficient to realize the logically equivalent condition,
i.e., the pin-projected MWIS configurations match the target constraint.
Weights $\Delta\in\mathbb{Z}_{\ge 0}^{|V|}$ are then assigned to vertices 
by solving the following integer linear program:

\begin{equation}\label{eq:lp}
    \begin{split}
        &\min_{\Delta \in \mathbb{Z}^{|V|}_{\geq 0}} \sum_i \Delta_i\\
        &\sum_i n_i' \Delta_i < \sum_i n_i \Delta_i, \quad \forall\, \mathbf{n} \in \mathcal{M}_{\text{min}},\, \mathbf{n}' \in \mathcal{M} \setminus \mathcal{M}_{\text{min}} \\
        &\sum_i n_i \Delta_i = \sum_i n_i' \Delta_i, \quad \forall\, \mathbf{n},\, \mathbf{n}' \in \mathcal{M}_{\text{min}}.
\end{split}
\end{equation} 
These two constraints ensure that all configurations in $\mathcal{M}_{\min}$ acquire identical energies, while every non-target configuration in $\mathcal{M} \setminus \mathcal{M}_{\min}$ is lifted above the ground-state energy. Optionally, one can minimize $\sum_v \Delta_v$ to keep the vertex weights small. Solving this integer program using a standard solver (e.g., Gurobi~\cite{gurobi}) yields feasible weight assignments, and iterating over the graph space $\mathcal{G}$ enables systematic discovery of gadgets that encode the target logical constraints.

For gadgets on unit-disk graphs such as TLSGs, the candidate graph space $\mathcal{G}$ is generated by applying a Boolean mask to the triangular lattice. 
The \emph{pin vertices} representing logical variables are constrained to lie on the boundary, ensuring that ancilla vertices do not form unwanted connections when gadgets are composed.

The gadgets used in \Cref{thm:main} obtained through this search framework are shown in Fig.~\ref{fig:show_gadget}(a-c). In particular, the crossing gadget on the triangular lattice requires 12 vertices with weights in the range 1 to 4. A brute-force construction as in Ref.~\cite{Nguyen2023} would necessitate enumerating at least $4^{12}  \approx 1.6 \times 10^7$ weight allocations, making the gadget-design computationally infeasible. By contrast, our integer-programming approach finds valid solutions within only a few seconds.

\begin{figure}[t]
    \centering
    \includegraphics[width=\linewidth]{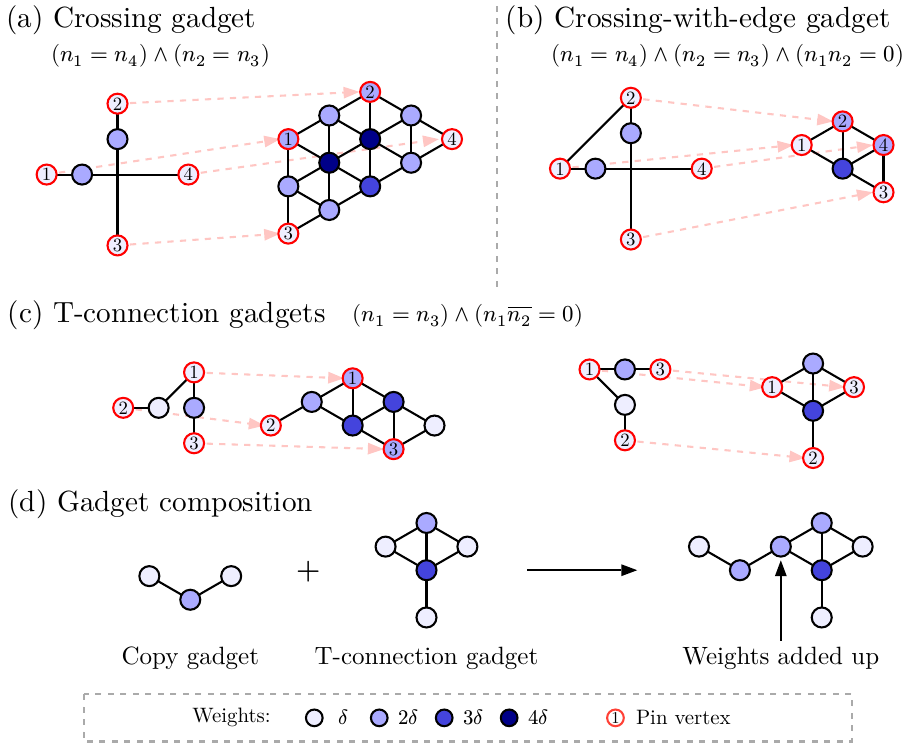}
    \caption{(a-c) Three essential gadgets for TLSG encoding. In each subfigure, the left column is the source graph, and the right column is the mapped graph on a triangular lattice. The red-framed vertices on the boundary are pin vertices, and only pin vertices can connect to external vertices. The full list of gadgets is provided in the Supplemental Material. (d) The composition of a copy gadget and a T-connection gadget. This requires summing the weights of the connected pin vertices at the junction.}
    \label{fig:show_gadget}
\end{figure}

\paragraph{\label{sec:result}Numerical simulation results}—
To benchmark the TLSG encoding, we simulate the quantum optimization process for the encoded $K_{2,3}$ instance shown in Fig.~\ref{fig:main}(d), comparing it against the KSG encoding obtained via \textit{UnitDiskMapping.jl}~\cite{UnitDiskMapping2025} shown in Fig.~\ref{fig:K23_pulse}(a). We quantify encoding performance using the \emph{violation rate}—the per-bond probability that independence constraints are violated in the final measurement outcomes (details in Supplemental Material).

\begin{figure}[htbp]
\centering
\includegraphics[width=
\linewidth]{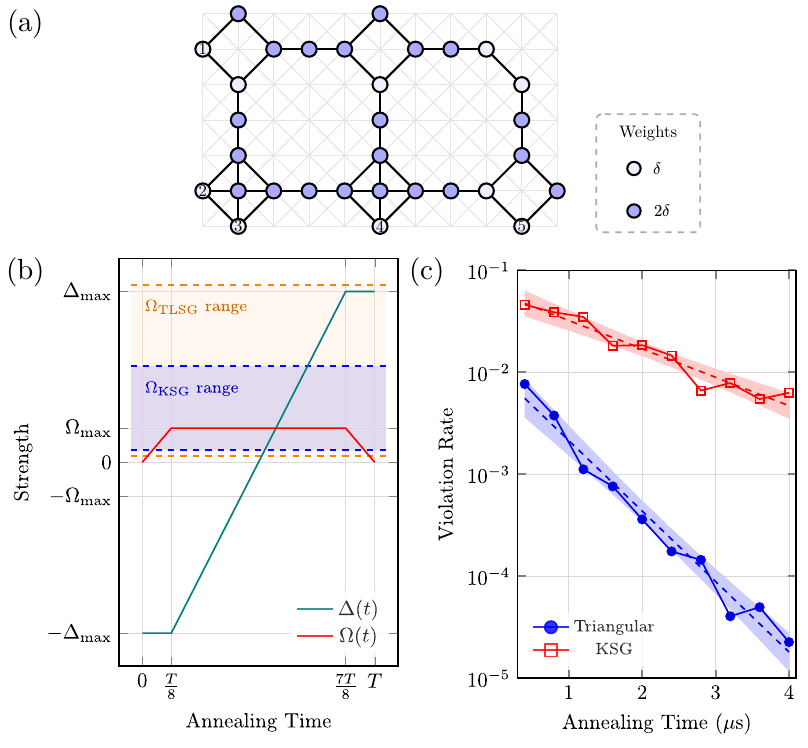}
\caption{Quantum annealing simulations comparing TLSG and KSG encodings for the $K_{2,3}$ problem.
(a) The KSG encoding of $K_{2,3}$ for comparison. The TLSG encoding is shown in Fig.~\ref{fig:main}(d).
(b) Piecewise-linear annealing pulses for Rabi frequency $\Omega(t)$ and detuning $\Delta(t)$. 
The shaded regions indicate the operational range $\frac{C_6}{R_{\rm min}^6} <\Omega <\frac{C_6}{r_{\rm max}^6}$ with $C_6=2\pi\times 862690 \,\text{MHz} \, \mu\text{m}^6$. The lattice unit is chosen so that $\Omega_{\rm max}$ equals the geometric mean of the interaction bounds.
(c) Violation rate of independence constraints versus total annealing time for both encodings. 
Dashed lines show exponential fits, and shaded regions denote $95\%$ confidence intervals. 
The TLSG encoding reduces the violation rate by nearly two orders of magnitude compared to the KSG encoding across most annealing regimes.}
\label{fig:K23_pulse}
\end{figure}

In quantum annealing, the system evolves adiabatically under time-varying Rabi frequency $\Omega(t)$ and detuning $\Delta(t)$. Following the experimental protocol of Ref.~\cite{Ebadi2022}, we employ unoptimized piecewise-linear ramps for both parameters. 
To ensure fair comparison, we adopt identical parameter strategies for TLSG and KSG instances: given preset maximum values $\Omega_{\rm max}$ and $\Delta_{\rm max}$, the lattice unit $a$ is chosen such that $\Omega_{\rm max}$ equals the geometric mean of interaction strengths at distances $R_{\rm min}$ and $r_{\rm max}$, i.e., $\Omega_{\rm max} = \sqrt{\frac{C_6}{R_{\rm min}^6} \cdot \frac{C_6}{r_{\rm max}^6}} = \frac{C_6}{\sqrt{R_{\rm min}^3 r_{\rm max}^3}}$.
We perform tensor network simulations across various annealing times (from $0.4\,\mu\text{s}$ to $4.0\,\mu\text{s}$) to analyze violation rate scaling.
For detailed parameter setup and simulation methods, see Supplemental Material.

The key advantage of TLSG encoding lies in its enlarged operational energy window. As illustrated in Fig.~\ref{fig:K23_pulse}(b), the annealing pulses must operate within constraints imposed by strong interactions at $r_{\rm max}$ (maintaining blockade) and weak interactions at $R_{\rm min}$ (avoiding spurious blockade of disconnected atoms).
The superior quality factor of TLSG encodings translates directly to enhanced performance. Figure~\ref{fig:K23_pulse}(c) demonstrates that TLSG encoding reduces violation rates by nearly two orders of magnitude compared to KSG encoding.
Furthermore, the violation rate exhibits steeper exponential decay with increasing annealing time, demonstrating more effective constraint enforcement by the TLSG encoding.
For annealing times exceeding $3\,\mu\text{s}$, the TLSG violation rate drops below $10^{-4}$, eliminating post-processing even for  state-of-the-art arrays with $6100$ atoms~\cite{Manetsch2024}, as $(1-10^{-4})^{6100} \sim 0.543$, allowing us to simply discard the illegal configurations without post-processing.
 
\paragraph{\label{sec:conclusion}Conclusions and outlook}
— Our work establishes triangular-lattice subgraphs (TLSGs) as superior to King's subgraphs for Rydberg quantum optimization. The enhanced quality factor ($Q = \sqrt{3}$ vs $\sqrt{2}$) reduces constraint violations by nearly two orders of magnitude, potentially eliminating the need for post-processing.
We contribute two key innovations: (i) the first systematic MWIS encoding scheme for TLSGs with $O(n^2)$ overhead using crossing-lattice constructions and specialized gadgets, and (ii) an automated integer-programming framework for discovering valid gadgets—previously computationally infeasible. Both are implemented in open-source Julia packages available online~\cite{UnitDiskMapping2025,GadgetSearch2025}.

This framework has the potential to be applied to other NP-hard problems (graph coloring, dominating sets, etc.) and alternative platforms (trapped ions~\cite{doi:10.1073/pnas.2006373117}, superconducting qubits~\cite{harrigan2021quantum}), broadening quantum optimization beyond Rydberg systems.

\paragraph{Acknowledgment}
— This work was partially supported by the National Key R\&D Program of China (Grant No.~2024YFB4504004), the National Natural Science Foundation of China under grant nos.~12404568, and the Guangzhou Municipal Science and Technology Project (No. 2024A03J0607). 
\bibliography{refs}

\clearpage
\onecolumngrid
\renewcommand{\theequation}{S\arabic{equation}}
\setcounter{equation}{0} 

\renewcommand{\thefigure}{S\arabic{figure}}
\setcounter{figure}{0} 
\renewcommand{\thetable}{S\arabic{table}}
\setcounter{table}{0}

\section*{Supplementary Material: Encoding computationally hard problems in triangular Rydberg atom arrays}

In this supplemental material, we provide detailed technical information and extended discussions that support the main results presented in the paper. 

Section~\ref{appsec:deform} introduces the triangular lattice geometry and explains the coordinate transformation convention used to map between the square-grid representation and the physical triangular lattice. Section~\ref{appsec:composition} presents the gadget composition principle, which forms the theoretical foundation for combining individual logical constraints into complex systems through vertex merging operations. 
Section~\ref{appsec:mapping} describes our systematic methodology for encoding arbitrary graphs onto triangular-lattice subgraphs (TLSGs)—unit-disk graphs laid out on a triangular lattice. Section~\ref{appsec:gadget} details the algorithmic framework for searching and constructing gadgets that encode specific logical constraints. Finally, Section~\ref{appsec:simulation} describes the numerical simulation setup used to benchmark the performance of our TLSG encoding against the established King's subgraph (KSG) approach through quantum annealing simulations.

\section{\label{appsec:deform}Triangular lattice geometry and layout convention}

In a TLSG with spacing $a$ (Fig.~\ref{fig:reshape}(b)), each vertex has six nearest neighbors at distance $a$, 
and the shortest non-adjacent distance is $\sqrt{3}\,a$, hence
\[
Q_{\mathrm{TLSG}}=\frac{R_{\min}}{r_{\max}}=\frac{\sqrt{3}\,a}{a}=\sqrt{3}.
\]

In a King's graph (KSG) each vertex has eight neighbors (four axial at distance $a$ and four diagonal at distance $\sqrt{2}\,a$; Fig.~\ref{fig:reshape}(a)). 
Here the largest edge length is $r_{\max}=\sqrt{2}\,a$, while the shortest non-edge distance is $R_{\min}=2\,a$, giving
\[
Q_{\mathrm{KSG}}=\frac{R_{\min}}{r_{\max}}=\frac{2\,a}{\sqrt{2}\,a}=\sqrt{2}.
\]

Since Rydberg interactions scale as $V(r)\propto 1/r^6$, the relevant interaction-scale separation is
\[
\frac{V(r_{\max})}{V(R_{\min})}=Q^6,
\]
so $Q_{\mathrm{TLSG}}^6=27$ versus $Q_{\mathrm{KSG}}^6=8$. In other words, the triangular lattice provides a substantially larger energetic gap between connected and non-connected atoms, enhancing the robustness of maximum weighted independent set (MWIS) encodings against spurious long-tail interactions.

In our code implementation, the triangular lattice is represented on a square grid to simplify coordinate assignment. However, this representation distorts the original geometry and requires coordinate transformation.

To recover the physical triangular lattice from the square-grid representation, 
we perform a coordinate transformation that accounts for both spacing anisotropy 
and parity-dependent offsets.  
In this convention, vertices in odd-numbered columns are shifted vertically by half a lattice spacing ($0.5\,a$), 
while the horizontal spacing between columns is scaled by a factor of $\sqrt{3}/2$.  
Formally, for an internal coordinate $(x, y)$, the physical coordinates $(X, Y)$ are
\begin{equation}
(X, Y) =
\left(
\frac{\sqrt{3}}{2} \, a \, x, \quad
a \left[ y + \frac{1}{2} (x \bmod 2) \right]
\right).
\end{equation}
\begin{figure}[!htbp]
    \centering
    \includegraphics[width=0.5\textwidth]{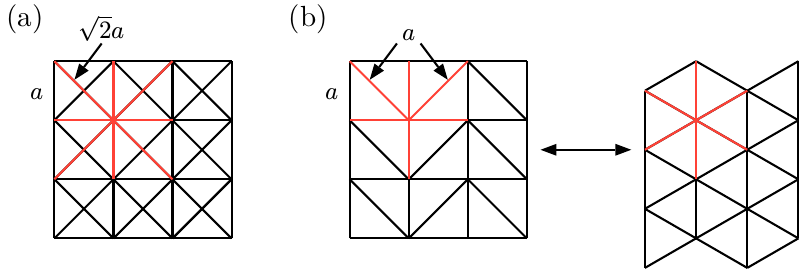}
    \caption{(a) King's subgraph (KSG) on a square lattice. (b) Triangular lattice subgraph (TLSG) and the transformation process between the original triangular lattice and its reshaped square-grid representation used in the code implementation.}
    \label{fig:reshape}
\end{figure}

As a consequence of this layout transformation, horizontal wires (aligned along rows) become zigzag-shaped in the true triangular lattice, while vertical wires (aligned along columns) remain straight. This distinction affects gadget design since the geometry and connectivity patterns differ between orientations.

\section{\label{appsec:composition}Gadget composition principle}
This section reviews and summarizes the gadget composition principle introduced by Nguyen et al.~\cite{Nguyen2023}. In the context of the MWIS problem, a NOT gadget can be represented by a simple graph consisting of just two vertices with identical weights and a single edge connecting them (Fig.~\ref{appfig:not_gadget}(a)). This construction ensures that only one of the two vertices can be included in any MWIS, effectively encoding the logical negation: if one vertex (representing a Boolean variable) is selected, the other (representing its negation) must be excluded, and vice versa. 

\begin{figure}[!htbp]
    \centering
    \includegraphics[width=0.45\textwidth]{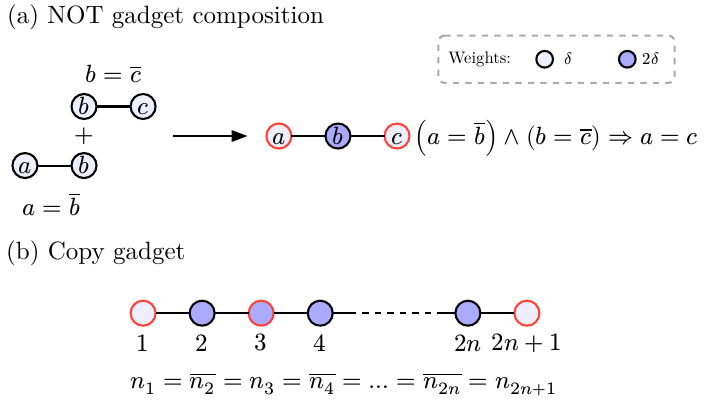}
    \caption{Illustration of how logical conjunction translates to gadget composition in the MWIS framework. (a) Two NOT gadgets encode the constraints $a = \bar{b}$ and $b = \bar{c}$, and are combined by merging the vertices representing the shared variable $b$, ensuring both constraints are enforced simultaneously. (b) shows an extended version where an odd-length vertex wire is constructed from a sequence of NOT gadgets. Red-framed vertices are equivalent in the MWIS sense, allowing local information to propagate over long distances.}
    \label{appfig:not_gadget}
\end{figure}

To illustrate how logical conjunction corresponds to graph composition in the MWIS framework, consider encoding the two constraints: $a = \bar{b}$ and $b = \bar{c}$, see \Cref{appfig:not_gadget}(a). Each of these constraints can be represented by a simple NOT gadget. 
To enforce both constraints simultaneously, the two gadgets must be composed. Since both constraints involve the same logical variable $b$, we must identify the two vertices representing $b$, i.e., merge $v_b$ and $v_b'$ into a single vertex. This merged vertex now participates in both constraints.

From the energy perspective, each gadget originally assigned a weight $\delta$ to its vertex $v_b$ (and likewise to $v_a$ and $v_c$). After merging, the shared vertex $v_b$ contributes to two energy terms—one from each NOT constraint. Since the MWIS energy function is additive over vertex weights, the new weight of $v_b$ must be updated to $2\delta$, reflecting the combined energetic influence of both constraints on that variable.

In summary, merging two gadgets on a shared logical variable enforces the conjunction of their respective constraints. 
The overlapping vertex inherits the sum of the weights from all participating gadgets, 
so that the total MWIS Hamiltonian remains a linear sum of individual gadget energies. 
Merging vertices does not affect the independence constraints, provided that no new edges are introduced, 
ensuring that each logical constraint contributes correctly to the total energy. 
As a result, the ground state of the composite system simultaneously satisfies all encoded constraints.
We summarize the above procedure in the following according to Ref.~\cite{Nguyen2023}.

\begin{proposition}[Gadget composition principle]
    \label{prop:composition}
    Let $G_1=(V_1,E_1,\Delta_1)$ and $G_2=(V_2,E_2,\Delta_2)$ be gadgets that are logically equivalent on their pin sets $P\subseteq V_1$ and $Q\subseteq V_2$ to logical constraints $\mathcal{L}_1$ and $\mathcal{L}_2$, respectively. Let $p\in P$ and $q\in Q$ represent the same logical variable. Form the composite graph $G$ by identifying $p$ and $q$ into a single vertex $v$ and inheriting edges from $E_1\cup E_2$ (updated to reflect the merge). Assume the MWIS objective is additive in vertex weights and that the merge introduces no edges beyond $E_1\cup E_2$. Define
    \[
    \Delta(v)=\Delta_1(p)+\Delta_2(q),
    \]
    and keep all other vertex weights unchanged. Let the composite pin set be $(P\setminus\{p\})\cup(Q\setminus\{q\})\cup\{v\}$.
    Then the MWIS solutions of $G$ are logically equivalent on the composite pin set to the conjunction $\mathcal{L}_1\wedge \mathcal{L}_2$.
\end{proposition}

The \emph{copy gadget}, considered in this work, is constructed from an even number of NOT gadgets according to \Cref{prop:composition}, 
which naturally results in an odd-length vertex wire of $2n+1$ atoms 
with boundary atoms assigned weight $\delta$ and interior atoms weight $2\delta$ (see \Cref{appfig:not_gadget}(b))~\footnote{All the gadgets are named according to their function, even though their specific implementations may differ. 
For example, in Ref.~\cite{Nguyen2023}, the crossing gadget enforces $(n_1 = \overline{n_4}) \land (n_2 = \overline{n_3})$, 
and the copy gadget contains an even number of vertices. 
In the present work, these gadgets may be implemented slightly differently—typically differing by a single NOT gadget per logical variable—but they perform analogous roles within the crossing lattice.}. 
Under the strong Rydberg blockade, ideally no two adjacent atoms are simultaneously excited. 
In the $\Omega/\Delta \rightarrow 0$ limit, the lowest-energy configurations of the odd-length wire form a two-fold degenerate $\mathbb{Z}_2$-ordered pair:
$|0\rangle_\ell = |0101\ldots010\rangle, \quad |1\rangle_\ell = |1010\ldots101\rangle$ \cite{bernien2017probing}.
These two logical states encode a binary variable indicating whether the corresponding graph vertex is excluded (0) or included (1) in the MWIS.

\begin{proposition}
    All odd-indexed vertices in a copy gadget are logically equivalent to the corresponding vertex in the original graph.
\end{proposition}

\section{\label{appsec:mapping}Graph encoding methodology: \\ From arbitrary graphs to UDGs on triangular-lattice subgraphs (TLSGs)}

Given an arbitrary simple graph $G = (V, E)$ with vertex set $V$ and edge set $E$, our objective is to construct an encoding into a weighted unit-disk graph on a TLSG (see Fig.~\ref{fig:reshape}) such that the maximum (weighted) independent set problem on the original graph can be reduced to the MWIS problem on the encoded graph.

To address this challenge, we introduce the \emph{crossing lattice} framework as in the KSG encoding~\cite{Nguyen2023}. In this approach, each logical vertex is represented as a vertex wire (copy gadget), arranged so that every pair of logical vertices intersects at exactly one crossing. This design propagates local adjacency relations along dedicated paths, ensuring that only pairwise interactions need to be physically realized at each intersection. Each edge in the original graph is represented by a crossing-with-edge gadget at intersections, while non-adjacent vertex pairs are either handled by crossing gadgets or remain unconnected.

We then present a systematic approach for encoding arbitrary graphs into TLSGs through the construction of crossing lattices. As described in \Cref{appsec:deform}, indices are assigned on a square-grid representation for convenience, whereas the physical coordinates and distances are consistently interpreted on the triangular lattice via the mapping convention.

\subsection{\label{appsubsec:copy}Step 1: Construct Crossing Lattice}
To construct the crossing lattice, we use L-shaped \emph{copy lines} as the geometric scaffold. 
Each copy line specifies the horizontal and vertical spans of an L-shaped wire assigned to a logical vertex. 
Arranged so that any two vertices intersect at most once, these wires faithfully reproduce the adjacency of the original graph. 
Finally, each copy line is populated with an odd-length copy gadget, turning the abstract scaffold into a physical vertex wire.

\begin{figure}[b]
    \centering
    \begin{tabular}{|c|c|c|c|c|}
        \hline
           Vertex  & Vertical slot & Vertical range & Horizontal slot & Horizontal range \\ \hline
        1 & / & / & 1 & [1, 4] \\ \hline
        2 & 2 & [1, 2] & 2 & [2, 4] \\ \hline
        3 & 3 & [1, 3] & 3 & [3, 4] \\ \hline
        4 & 4 & [1, 4] & / & / \\ \hline
    \end{tabular}
    \end{figure}

The complete graph $K_4$ provides a concrete illustration of this procedure.
As shown in Fig.~\ref{fig:crossinglattice}(a) and the table below, different colored copy lines determine the placement of vertex wires and their intersections, and physical sites are then populated along each line.
These slots partition the grid into cells, with size $s=6$ for the TLSG and $s=4$ for the KSG.
After vertices are placed at the subdivided grid points, each copy line forms an odd-length copy gadget.

At each crossing between the copy gadgets of vertices $u$ and $v$, we introduce an edge if $(u,v)\in E$.
In the Hamiltonian, such an edge corresponds to a hard constraint term $\infty \, n_u n_v$.

\begin{figure}[htbp]
    \centering
    \includegraphics[width=0.9\linewidth]{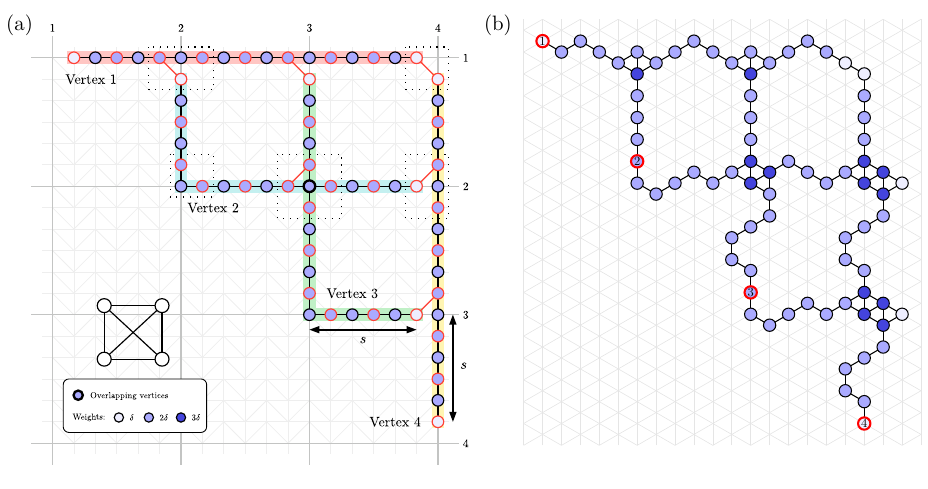}
    \caption{(a) Construction of the crossing lattice from copy lines for the complete graph $K_4$ (before gadget replacement). 
    Dark gray lines denote copy line slots, each colored wire represents a vertex of the original graph, red-framed sites mark the vertices logically equivalent to the originals, and red edges encode the original graph’s connectivity.
    Dashed boxes highlight some regions violating TLSG constraints. 
    (b) The crossing lattice shown in physical coordinates after gadget replacement and fine-tuning. 
    The number of auxiliary atoms can be further optimized, for example by manually removing dangling legs.}
    \label{fig:crossinglattice}
\end{figure}

\subsection{Step 2: Gadget replacement}

\paragraph{Identify and replace invalid substructures.}
The dashed boxes in Fig.~\ref{fig:crossinglattice}(a) highlight local substructures of the crossing lattice that violate UDG constraints in the TLSG after Step 1. These regions are extracted and replaced with logically equivalent gadgets, determined by the logical rules of the boundary pin vertices. When performing replacement, the weights of pin vertices connected to external structures must be updated according to Proposition~\ref{prop:composition}. Details of the replacement are shown in Fig.~\ref{fig:gadget_replacement}.

\paragraph{Apply geometric fine-tuning.}
Replacement gadgets often differ in size and pin placement from the originals, so additional adjustments are required (see Fig.~\ref{fig:gadget_replacement}). Fine-tuning consists of bending, extending, or shortening copy gadgets to maintain connectivity and ensure consistent pin-vertex alignment. Wires may need to detour sideways to preserve the parity of each wire. The resulting layout after replacement and fine-tuning is shown in Fig.~\ref{fig:crossinglattice}(b).

\paragraph{Adjust vertex weights with detuning shifts.}
If an original vertex $v_i$ carries a weight $w_i$, this can be renormalized and transferred to the encoding by adding a small detuning shift $\varepsilon_i$ at the corresponding logically equivalent vertex, with
$\varepsilon_i \propto w_i\, $, subject to the requirement that it is large enough to lift unwanted degeneracies but small enough not to disturb the $\delta$-scale constraints. Even for unweighted MIS instances, a small nonzero $\varepsilon$ is added to break unwanted degeneracies. For a single copy gadget, such a detuning acts as a symmetry-breaking perturbation: it lifts the degeneracy between the two MWIS states without affecting the $\mathbb{Z}_2$ ordering of the wire.

\paragraph{Optimize for compactness.}
Although the encoding is complete after gadget replacement and fine-tuning, further optimization can reduce resource overhead. Dangling vertex wires may be removed, and copy gadgets shortened, provided that parity constraints are respected. At present, such optimization needs to be performed manually.

\paragraph{Finalize by reshaping to the triangular lattice.}
The last step maps the square-grid representation onto the physical triangular lattice using the procedure in Section~\ref{appsec:deform}. This reshaping preserves all unit-disk edges and requires no further adjustments. For readout, the measurement outcome may be taken from any logically equivalent vertex along a vertex wire.

\begin{figure*}[htbp]
    \centering
    \includegraphics[width=\textwidth]{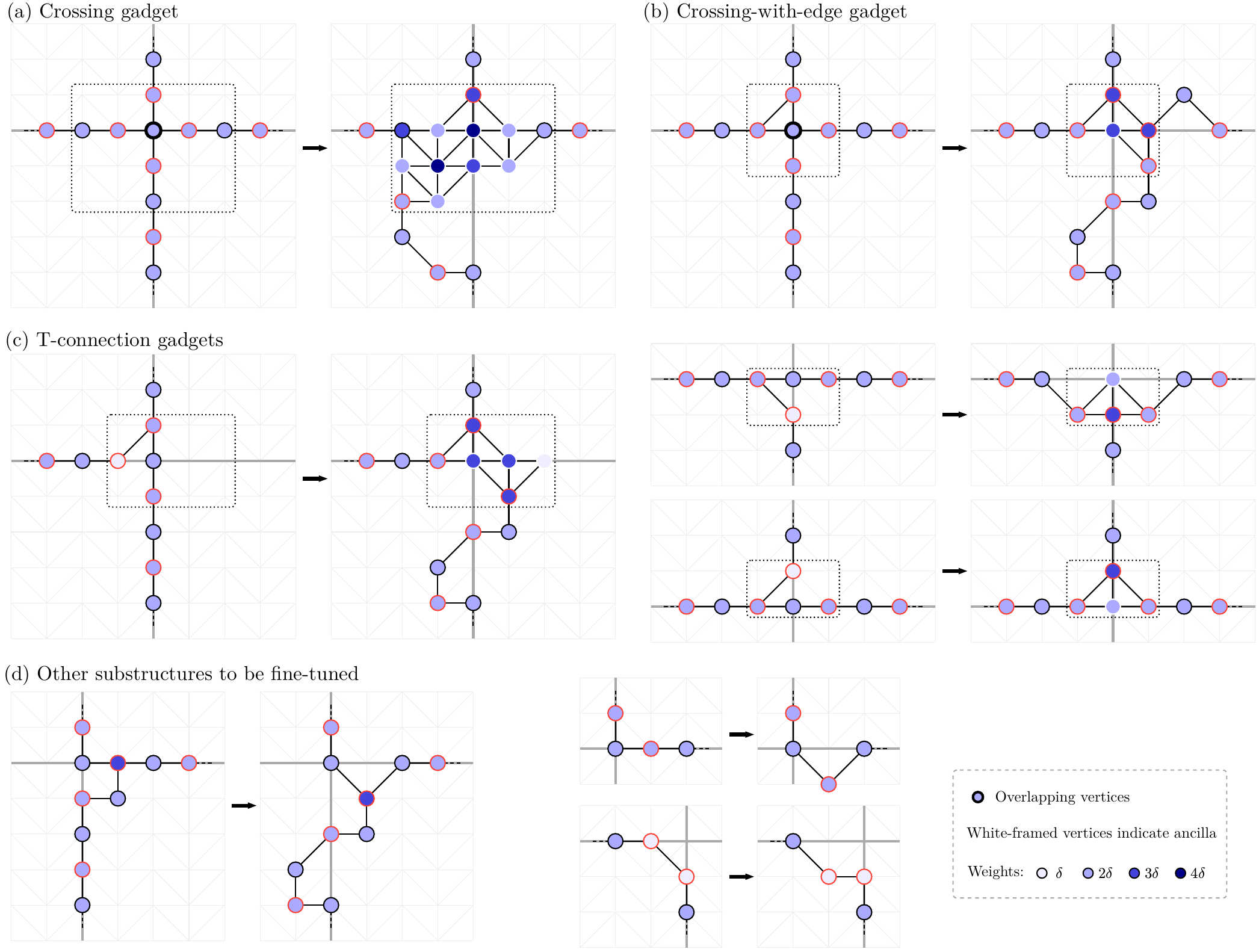}
    \caption{Illustration of the gadget replacement and fine-tuning procedure. The dashed boxes indicate regions where gadgets are directly replaced. Fine-tuning is applied after replacement to ensure proper integration with the existing structure, subject to two constraints: (1) the parity of the number of atoms on each vertex-wire is preserved, and (2) no additional connections are introduced beyond those that already exist. Red-framed vertices indicate the logical counterparts on the corresponding vertex wires, and white-framed vertices inside each gadget have no logical role and function solely as ancillas.}
    \label{fig:gadget_replacement}
\end{figure*}

\subsection{Overhead Analysis}
For a graph with $n$ vertices and $m$ edges, the TLSG encoding scheme produces an encoding containing $N_\text{TLSG}$ vertices, given approximately by
\begin{equation}
\begin{split}
N_\text{TLSG} &\lesssim 6n(n-1) + 4m_\text{right} 
+ 1\left( m - m_\text{top} - m_\text{right} \right) 
+ 4 \left( \frac{n(n-1)}{2} - m \right) \\
&\sim 8n^2 - 6n - 3m,
\end{split}
\end{equation}
where $6$ is the size of the cross-lattice cell $s$ (see \Cref{fig:crossinglattice}(a)), $4$ the ancilla vertex count of the rightmost T-connection gadget, $1$ that of the crossing-with-edge gadget, $4$ that of the crossing gadget, $m_\text{right}, m_\text{top} \leq n-1$ denote the numbers of edges intersecting the rightmost and topmost copy lines, respectively. 

This analysis indicates that our method requires at most $O(8n^2)$ atoms, whereas the encoding overhead on the KSG is at most $O(4n^2)$~\cite{Nguyen2023,schuetz2025quantum}. This estimate is rather coarse, as it does not account for potential optimizations such as vertex reordering~\cite{Nguyen2023} after encoding.

\begin{figure*}[htbp]
    \centering
    \includegraphics[width=0.95\linewidth]{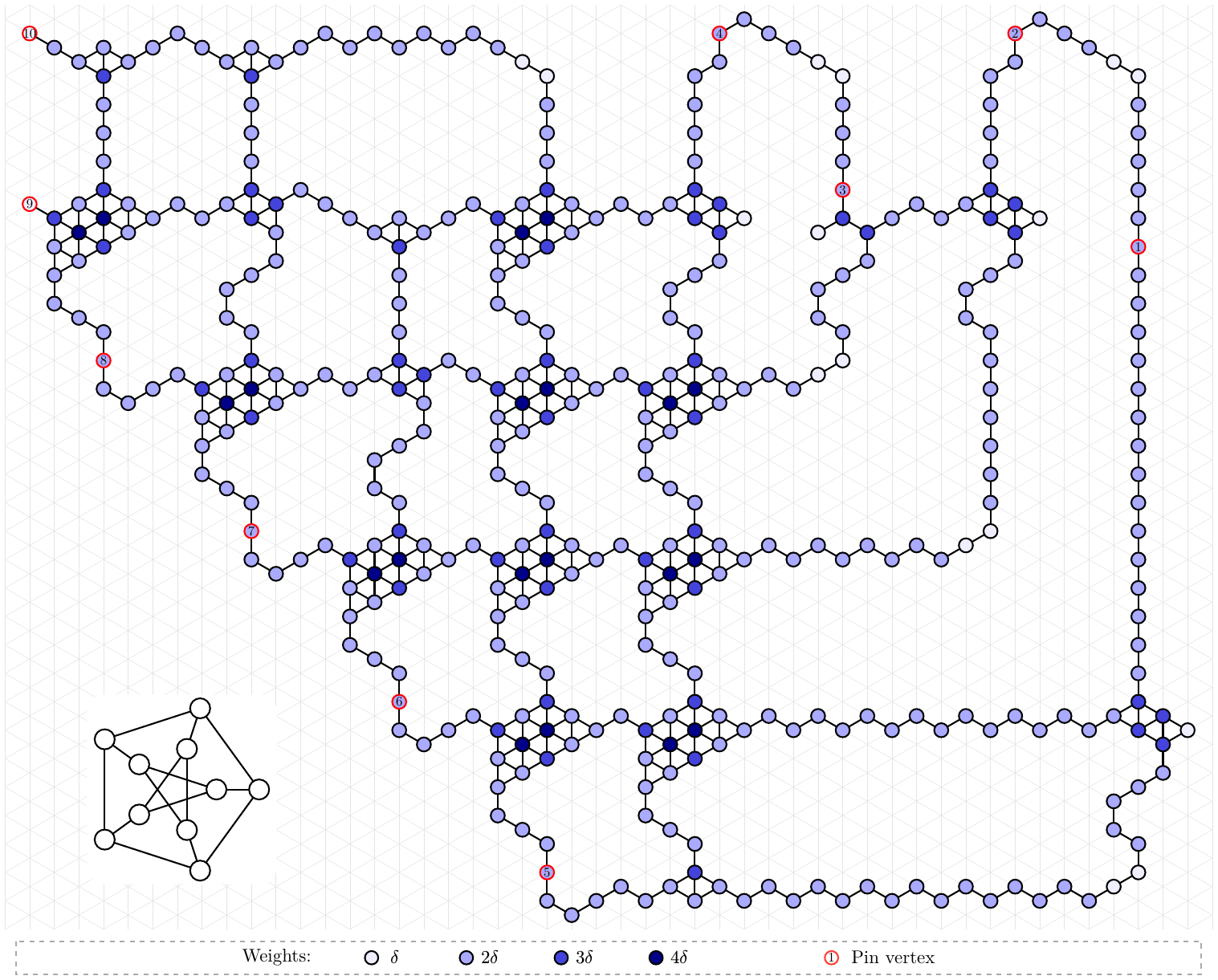}
    \caption{Encoding of the Petersen graph onto a TLSG. Vertex color indicates weight. 
    Vertices in the original graph correspond to the red-framed vertices. 
    Membership in the original graph's independent set is inferred by projective measurement of the corresponding red-framed vertex in the lattice.}
    \label{fig:show_petersen}
\end{figure*}

\subsection{Code Availability and Usage}
We build upon the existing Julia package \textit{UnitDiskMapping.jl}~\cite{UnitDiskMapping2025}, which reduces generic maximum (weighted) independent set, QUBO, or integer factorization problems to KSG formulations. In this work, we extend the package by introducing a new module that encodes generic MIS/MWIS instances directly onto TLSGs, enabling natural implementation on neutral-atom quantum computers. The updated codebase, along with usage examples and documentation, is openly available at~\cite{UnitDiskMapping2025}.

Here, we take the complete graph $K_4$ as an example to demonstrate how to use the code. The entire encoding workflow, though conceptually involving many steps, has been fully encapsulated, so that a single function call suffices to perform the whole process.

\begin{lstlisting}[style=JuliaREPL]
    julia> using UnitDiskMapping, Graphs

    julia> graph = complete_graph(4)
    {4, 6} undirected simple Int64 graph

    julia> triangular_weighted_res = map_graph(TriangularWeighted(), graph);
\end{lstlisting}

The result can then be visualized, producing plots such as those shown in Fig.~\ref{fig:crossinglattice}(b). A larger example using the Petersen graph is shown in Fig.~\ref{fig:show_petersen}.

\begin{lstlisting}[style=JuliaREPL]
    julia> using LuxorGraphPlot.Luxor, LuxorGraphPlot

    julia> show_grayscale(triangular_weighted_res.grid_graph)   # Visualize weighted graph with colors

    julia> show_pins(triangular_weighted_res) # Highlight pin vertices in the encoding graph
\end{lstlisting}

The code also provides APIs to verify the correctness of the encoding results; for detailed usage, please refer to the documentation.

\section{\label{appsec:gadget}Details of the search algorithm}
In this section, we describe how, given a logical rule $\mathcal{L}$, one can construct or search for a graph $G = (V,E)$ along with positive vertex weights $\{\Delta_v \mid v \in V\}$ such that the solution to the MWIS problem on this graph corresponds to the given logical rule. Such a graph is what we previously referred to as a gadget which maps a discrete set of states to the energy-minimizing solution of a combinatorial optimization problem.
\subsection{Preliminaries}
\begin{definition}[Maximal Independent Set]
A maximal independent set is an independent set that is not a proper subset of any other independent set.
\end{definition}

In other words, a maximal independent set $S$ satisfies the following condition: adding any vertex not in $S$ would violate its independence, i.e.,
\begin{equation}
    \forall v \in V \setminus S, \; \exists u \in S \text{ such that } (u,v) \in E.
\end{equation}

\begin{proposition}[MWIS is maximal]
    For any positive-weighted graph $G = (V, E, \Delta)$, a maximum weighted independent set is always a maximal independent set.
\end{proposition}
    
\begin{proof}
Suppose there exists a maximum weighted independent set $S$ that is not maximal. Then there exists a vertex $u \in V \setminus S$ that is not adjacent to any vertex in $S$. In this case, $S \cup \{u\}$ forms a larger independent set with total weight $\Delta(S) + \Delta(u) > \Delta(S)$, which contradicts the assumption that $S$ is maximum weighted. Therefore, a maximum weighted independent set must be maximal.
\end{proof}

As we will see in the following sections, the concepts of maximal and MWIS form the foundation for constructing gadgets: by carefully selecting maximal independent sets and assigning vertex weights, we can encode logical rules into the MWIS problem on unit-disk graphs.



\subsection{Search algorithm}

To find gadgets that satisfy the desired properties within a constrained graph dataset $\mathcal{G}$ (e.g., TLSGs), we design a systematic search procedure, as outlined in \Cref{alg:search-graph-logic}.

\begin{figure}[t]
    \centering
        \begin{algorithm}[H]
            \caption{Graph Traversal and Logic Constraint Matching Search Process} 
            \label{alg:search-graph-logic}
            \SetKwInOut{Input}{Input}
            \SetKwInOut{Output}{Output}
            \Input{Graph set to search $\mathcal{G}$; logic rules $\mathcal{L}$ (e.g., truth table of a Boolean function)}
            \Output{Set $\mathcal{S}$ of graphs and corresponding weight configurations that satisfy $\mathcal{L}$}
            
            $\mathcal{S} \gets \varnothing$ \tcp*{Initialize the solution set as empty}
        
            \ForEach{$G \in \mathcal{G}$}{
        
                $\mathcal{M} \gets$ \texttt{MaximalIndependentSets}$(G)$ \tcp*{Compute maximal independent sets}
                $O \gets$ \texttt{GetOpenNodes}$(G)$ \tcp*{Get candidate nodes that can serve as pins}
                
                \ForEach{$\text{pins} \in$ \texttt{GeneratePinConfigs}$(O,\mathcal{M}, \mathcal{L})$} { 
                    $\mathcal{M}_\text{min} \gets$ \texttt{GetProperMIS}$(\mathcal{M},\text{pins}, \mathcal{L})$ \tcp*{Select MISs that satisfy the logic rule}
                    \If{$\mathcal{M}_\text{min}$ is empty}{
                        continue
                    }
        
                    $\mathrm{IP} \gets$ \texttt{FormulateIP}$(\mathcal{M}, \mathcal{M}_\text{min},\text{pins}, \mathcal{L})$ \tcp*{Formulate a mixed-integer program}
                    
                    \If{\texttt{Solve}$(\mathrm{IP})$ is feasible}{
                        $\Delta \gets$ \texttt{ExtractWeights}$(\mathrm{IP}.\text{solution})$ \tcp*{Extract vertex weights}
                        
                        $\mathcal{S} \gets \mathcal{S} \cup \{(G, \text{pins}, \Delta)\}$ \tcp*{Record the current solution}
                    }
                }
            }
        
            \Return{$\mathcal{S}$}
        \end{algorithm}
\end{figure}

\textit{Step 1: Solving maximal independent sets $\mathcal{M}$}
For each candidate graph $G = (V,E) \in \mathcal{G}$, all maximal \emph{cliques} in the complement graph $\overline{G}$ can be found using the Bron-Kerbosch algorithm~\cite{bron1973algorithm}, which corresponds to all maximal independent sets $\mathcal{M}$ on $G$. The set of maximal independent sets $\mathcal{M}$ can be represented as a $|\mathcal{M}| \times |V|$ binary matrix, where each row corresponds to a maximal independent set $\mathbf{n} \in \{0,1\}^{|V|}$ and each column corresponds to a vertex $v_i$ in the graph.

\textit{Step 2: Selecting target maximal independent sets $\mathcal{M}_\text{min} \subseteq \mathcal{M}$ based on $\mathcal{L}$}
In our encoding scheme, pin vertices—representing logical variables at gadget boundaries—are chosen from the open boundary vertices. 
We first identify the candidate set of pin vertices $O$. 
Given the logical rule $\mathcal{L}$ and the set of maximal independent sets $\mathcal{M}$, the algorithm then enumerates all possible ordered pin configurations within $O$, which determine how logical bits are mapped onto gadget vertices.

Let the target configuration set for $\mathcal{L} = \{\mathbf{s}_1, \mathbf{s}_2, \dots, \mathbf{s}_N\}$, where each $\mathbf{s}_i \in \{0,1\}^{k}$, $k \leq |V|$. For a fixed pin configuration, a maximal independent set is said to satisfy $\mathbf{s}_i$ if its configuration values on the pin vertices exactly match $\mathbf{s}_i$.

Accordingly, for each target configuration $\mathbf{s}_i$, one can define the candidate set of maximal independent sets $\mathcal{M}_{\mathbf{s}_i} \subseteq \mathcal{M}$. If $\mathcal{M}_{\mathbf{s}_1}, \dots, \mathcal{M}_{\mathbf{s}_N}$ are all non-empty, then $\mathcal{M}_\text{min}$ is non-empty, indicating that the vertex configuration \emph{possibly} encodes the logical rule $\mathcal{L}$. As the number of ancilla vertices increases, the number of maximal independent sets also grows, so the selection of $\mathcal{M}_\text{min}$ satisfying $\mathcal{L}$ is generally not unique. 

\textit{Step 3: Formulating the integer programming problem}
For any vertex configuration with a non-empty $\mathcal{M}_\text{min}$, the algorithm formulates the integer programming problem $\mathrm{IP}$ defined in main text:
\begin{equation}
    \begin{split}
        &\min_{\boldsymbol{\Delta} \in \mathbb{Z}^{|V|}_{\geq 0}} \sum_i \Delta_i\\
        &\sum_i n_i' \Delta_i < \sum_i n_i \Delta_i, \quad \forall\, \mathbf{n} \in \mathcal{M}_{\text{min}},\, \mathbf{n}' \in \mathcal{M} \setminus \mathcal{M}_{\text{min}} \\
        &\sum_i n_i \Delta_i = \sum_i n_i' \Delta_i, \quad \forall\, \mathbf{n},\, \mathbf{n}' \in \mathcal{M}_{\text{min}}.
    \end{split}
\end{equation}
As stated in the main text, vertex weights should be non-negative. In particular, if an optimal solution assigns a weight of zero to a vertex, it implies that the corresponding vertex in the resulting gadget can be removed.

In conclusion, the constraints in the IP are determined by the graph topology, the pin choice, and the logical rule. If the IP has a feasible solution—i.e., there exists a set of vertex weights $\Delta$ such that the graph under this configuration correctly encodes $\mathcal{L}$—then the graph, pin choice, and vertex weights are considered a valid solution and stored in the final solution set.

\begin{table}[htbp]
    \centering
    \begin{threeparttable}
    \caption{Examples of gadgets on the triangular lattice}
    \label{tab:logic-mwis}
    
    \begin{tabular}{c c c c @{\hskip 1cm} c c c c}
    \toprule
    \makecell{Logic \\ gate} & Gadget & \makecell{Pin \\ configuration} & \makecell{MWIS \\ energy} &
    \makecell{Logic \\ gate} & Gadget & \makecell{Pin \\ configuration} & \makecell{MWIS \\ energy} \\
    \midrule
    AND & 
    \parbox[c][2.6cm][c]{3cm}{\centering\includegraphics[width=1.8cm]{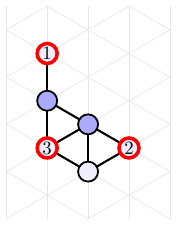}} &
    \texttt{000}, \texttt{010}, \texttt{100}, \texttt{111} & 3 &
    XOR &
    \parbox[c][3cm][c]{3cm}{\centering\includegraphics[width=1.8cm]{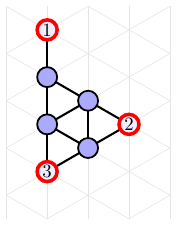}} &
    \texttt{000}, \texttt{011}, \texttt{101}, \texttt{110} & 4 \\ \hline
    
    NAND & 
    \parbox[c][2.6cm][c]{3cm}{\centering\includegraphics[width=2.1cm]{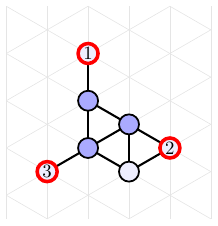}} &
    \texttt{001}, \texttt{011}, \texttt{101}, \texttt{110} & 4 &
    NOR &
    \parbox[c][2cm][c]{3cm}{\centering\includegraphics[width=2.1cm]{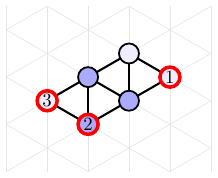}} &
    \texttt{001}, \texttt{010}, \texttt{100}, \texttt{110} & 3 \\ \hline
    
    OR & 
    \parbox[c][3cm][c]{3cm}{\centering\includegraphics[width=1.6cm]{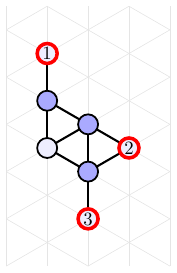}} &
    \texttt{000}, \texttt{011}, \texttt{101}, \texttt{111} & 4 &
     & & & \\ 
    \bottomrule
    \end{tabular}
    
    \begin{tablenotes}
    \footnotesize
    \item \includegraphics[width=0.5\linewidth]{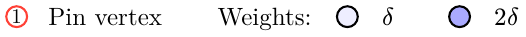} 
    \end{tablenotes}
    \end{threeparttable}
    \end{table}

\subsection{Code Availability and Usage}
We developed a Julia package, \textit{GadgetSearch.jl}~\cite{GadgetSearch2025}, for systematically searching gadgets in graph structures. The package provides tools to find weighted graphs whose MWIS solutions encode arbitrary logical constraints. It is suitable for TLSGs in this work.

Using this systematic search method, \Cref{tab:logic-mwis} shows examples of gadgets on TLSGs that implement the logical rules of AND, NAND, OR, NOR, and XOR gates. These gadgets are not unique. The code used to generate these gadgets is provided below; for detailed usage, please refer to the documentation.

\begin{lstlisting}[style=JuliaREPL]
    julia> using GadgetSearch, Gurobi, Combinatorics
    
    julia> const env = Gurobi.Env() # can be replaced with other solvers
    
    julia> truth_table = BitMatrix.([
        [0 0 0; 1 0 1; 0 1 1; 1 1 1],   # OR
        [0 0 0; 1 0 0; 0 1 0; 1 1 1],   # AND
        [0 0 1; 1 0 1; 0 1 1; 1 1 0],   # NAND = not(AND)
        [0 0 1; 1 0 0; 0 1 0; 1 1 0],   # NOR  = not(OR)
        [0 0 0; 1 0 1; 0 1 1; 1 1 0]    # XOR
    ])

    julia> generate_full_grid_udg(Triangular(), 2, 2; path="dataset.g6") 

    julia> dataloader = GraphLoader("dataset.g6") # you can provide your own graph datasets

    julia> results, failed = search_by_truth_tables(
        dataloader, 
        truth_table;
        optimizer=Gurobi.Optimizer, 
        env=env, 
        pin_candidates=collect(Combinatorics.combinations(1:4, 3)),  
        objective=x->sum(x)
    )
\end{lstlisting}

\section{\label{appsec:simulation}Numerical simulation setup}

To quantitatively evaluate the performance of our TLSG encoding, we simulate quantum annealing dynamics and benchmark against the established KSG encoding for the mapped $K_{2,3}$ graph. The simulations are designed to closely replicate experimental conditions in Ref.~\cite{Ebadi2022} while remaining computationally tractable. We employ the time-dependent variational principle (TDVP)~\cite{JuthoTDVP} with a matrix product state (MPS) ansatz.

Annealing performance is quantified by the \textit{violation rate} $p_\text{v}$, defined from the final measurement outcomes of $\ket{\psi_f}$ as
\begin{equation}
p_\text{v} \equiv \frac{1}{|E|} \sum_{\mathbf{n} \in \{0,1\}^{|V|}} P(\mathbf{n}) \sum_{(u,v) \in E} n_u n_v, \quad P(\mathbf{n}) = |\braket{\mathbf{n}}{\psi_f}|^2,
\label{eq:violation_rate}
\end{equation}
where $\mathbf{n}=(n_1,\dots,n_{|V|})$ denotes a bit-string configuration with $n_i\in\{0,1\}$ indicating whether vertex $v_i\in V$ is excited. Smaller $p_\text{v}$ corresponds to stronger enforcement of independence constraints.

The annealing protocol follows standard experimental practice~\cite{Ebadi2022}, employing an unoptimized piecewise-linear pulse (see main text Fig.~3(b)). All atoms are initialized in the ground state. The Rabi frequency $\Omega(\tau)$ ramps from 0 to $\Omega_\text{max}$ and back to 0, while the detuning $\Delta(\tau)$ sweeps from $-\Delta_\text{max}$ to $\Delta_\text{max}$ over the total annealing time $\tau \in [0,T]$. Key parameters are:
\begin{itemize}
\item Maximum Rabi frequency: $\Omega_\text{max} = 2\pi \times 4\,\text{MHz}$,
\item Maximum detuning: $\Delta_\text{max} = 5 \times \Omega_\text{max}$.
\end{itemize}

The blockade radius $R_b$ is defined as the distance at which the van der Waals interaction equals the maximum Rabi frequency, $C_6 / R_b^6 = \Omega_\text{max}$ with $C_6=2\pi\times 862690 \,\text{MHz}~\mu\text{m}^6$. To map the Rydberg blockade onto the unit-disk graph representation, we choose the lattice spacing $a$ (also the unit-disk radius) such that
\begin{equation}
R_b = \sqrt{R_{\min} r_{\max}},
\end{equation}
where $R_{\min}$ and $r_{\max}$ are the minimum distance between disconnected atoms and the maximum distance between connected atoms in the unit-disk graph. This sets the blockade radius at the geometric mean of these bounds, ensuring blockade for connected pairs while avoiding spurious long-range interactions~\cite{bloqade2023quera}. For the two lattice types considered, we have:
\begin{align*}
&\text{KSG:} && R_{\min}=2\,a, \quad r_{\max}=\sqrt{2}\,a \quad \Longrightarrow \quad a = 2^{-3/4} R_b, \\[2mm]
&\text{TLSG:} && R_{\min}=\sqrt{3}\,a, \quad r_{\max}=a \quad \Longrightarrow \quad a = 3^{-1/4} R_b.
\end{align*}

Time evolution is simulated using the time-dependent variational principle (TDVP)~\cite{JuthoTDVP} with a matrix product state (MPS) ansatz, a maximum bond dimension $\chi = 256$, and a time step of $0.1\,\Omega_\text{max}^{-1} \approx 0.004\,\mu\text{s}$. To accurately capture non-local interactions and geometric frustration in the Rydberg Hamiltonian, we employ a two-site TDVP scheme with third-order global subspace expansion (GSE)~\cite{GSETDVP}. After the bond dimension is saturated, we switch to one-site TDVP without GSE for computational efficiency.
\end{document}